\documentclass[11pt]{amsart}
\usepackage{enumerate}
\usepackage{graphicx,graphics}
\usepackage{amsfonts}
\usepackage{amssymb}
\usepackage{amsthm}
\usepackage{amsmath}
\usepackage{mathrsfs}
\usepackage{color}
\usepackage{algorithm}
\usepackage{algorithmic}
\input{xy}
\xyoption{all}

\newtheorem{theorem}{Theorem}[section]

\newtheorem{remark}[theorem]{Remark}

\newtheorem{proposition}[theorem]{Proposition}
\newtheorem{example}[theorem]{Example}
\newtheorem{definition}[theorem]{Definition}

\numberwithin{equation}{section}

\begin{document}

\title[Graph distances for determining entities relationships]{Graph distances for determining entities relationships: a topological approach to fraud detection}

\author[J.M. Calabuig, H. Falciani, A. Ferrer, L.M. Garc\'{\i}a and E.A. S\'anchez]{J.M. Calabuig, H. Falciani, A. Ferrer-Sapena, L.M. Garc\'{\i}a-Raffi and E.A. S\'anchez-P\'erez}

\address{
Instituto
Universitario de Matem\'atica Pura y Aplicada,  Universitat
Polit\`ecnica de Val\`encia\\ Camino de Vera s/n, 46022 Valencia. Spain.}

\email{jmcalabu@upv.es, rvfalciani@gmail.com, anfersa@upv.es,
lmgarcia@upv.es,
easancpe@upv.es}

\subjclass[2010]{Primary: 05C12; Secondary: 54A10}

\keywords{Graph distance, fraud detection, quasi-pseudo-metric, density, mass concentration, model}


\maketitle

\begin{abstract}
Given a set $\Omega$ and a proximity function $\phi: \Omega \times \Omega \to \mathbb R^+$, we define a
new metric for $\Omega$ by considering a path distance in $\Omega$, that is considered as a complete graph. We analyze the properties of such a distance, and several procedures for defining the initial proximity matrix $( \phi(a,b) )_{(a,b) \in \Omega \times \Omega}.$
Our motivation has its roots in the current interest in finding  effective algorithms for detecting and classifying relations among elements of a social network. For example, the analysis of a set of companies working for a given public administration or other figures in which automatic fraud detection systems are needed. Using this formalism, we state our main idea regarding fraud detection, that is founded in the fact that fraud can be detected because it produces a meaningful local change of density in the metric space defined in this way.
\end{abstract}

%
%
%
%
%

%

\section{Introduction}

The great increase in relations between financial actors due to the widespread use of social networks has opened the door to a new form of social organization, which can generate solid and powerful structures for committing financial fraud.
In parallel, the development of the same technical tools that permit to establish these criminal networks allow to  create new procedures to detect them. Indeed,
fraud detection is a current hot topic appearing daily in the news, and this produces a high demand of theoretical and practical mathematical instruments for fighting against fraud.

Since the mid-twentieth century, a consistent theoretical framework has been developed from the social sciences. The most powerful approach from this point of view seems to be the so called Fraud Triangle theory, that have show to be useful also in applications (see for example \cite{dorminey,mansor,trompeter}). This theoretical framework makes it possible to understand that fraud has its roots in psychological events determined by cultural structures and must finally be understood as a sociological fact (see diagram below). Therefore, social relations are those that make fraud processes visible, and can be formally analyzed using appropriate models that have to use large amounts of information and relational mathematical tools (\cite{mock,trompeter2}). Both are available today, so we are prepared to bring all these elements together to build general theoretical developments as well as computational tools for fraud detection.

\vspace{1cm}

\xymatrix{
\textbf{FRAUD TRIANGLE:}\,\,\,\ar[rr]^{\,\,\,\,\,\color{red}{Justification}}
\ar@{.>}[d]_{{\textit{\color{red}{{Motivation}}}} } & &  \ar@{.>}[d]_{\color{red}{Realization}}\textbf{\color{blue}{Rationalization}} \\
\textbf{\color{blue}{Social pressure}} \ar@{.>}[rr]^{\textit{\color{red}{Financial needs\,\,\,\,\,\,\,\,}}}& & \textbf{{\color{blue}Opportunity:} FRAUD}
}

\vspace{1cm}

The aim of this paper is to explain a new topological approach to understanding and detecting the processes of fraud. As we have already pointed out, the big amount of information that the new technologies bring into the scene have changed the way a scientist can understand the fraud as a mathematical phenomenon: invoices, emails, company registers, provide highly meaningful information that may help the analyst to detect  evidences of fraud. The extraordinarily large set of data that accompanies any fraud process makes necessary to change the usual analysis tools, traditionally based on the  study of the lawyers and the analysis of economists of related documents. New ways of understanding and new computing tools are clearly needed, and the theoretical development of the associated mathematical models has to grow together. Therefore, our idea is to propose a new model based on a topological graph approach to the analysis of networks.

Several mathematical theories have been already applied to fraud detection, involving quite different approaches: data science \cite{ngai,wang,zhao},    game theory \cite{wilks}, statistical analysis and machine learning \cite{ab,perols} and graph theory \cite{szarnyas} are some of them (see also  \cite{bolton,rich}). One of the most successful theories has proven to be the graph-based analytical approach, which has already given some programs for fraud detection, as Neo4j.  In this paper we propose a new technique for defining a model by means of quasi-pseudo-metrics for  complete graph structures. The vertices/nodes are the elements that have to be analyzed: persons, entities, companies, invoices or emails, for instance. Starting with a graph  with edges among vertices  having a finite set of properties, we establish a way for defining a family of quasi-pseudo-metrics for translating  the graph to a topological space.  To facilitate the explanation of the model, we will simplify our ideas in this paper  by assuming some requirements to ensure that the final quasi-pseudo-metric is in fact a metric.
 We will call such a structure a ``metric graph", and the topology will be constructed using quasi-pseudo-metrics (see for example \cite{kunzi,reilly}
for the basics). Once we can define  neighborhoods of vertices, we use the topological properties to characterize the relevant elements of the space, that have to become the main objects of the anti-fraud analysis. Besides the topological space, we need an additive set function acting in the class of all subsets of the original set of nodes ---a measure---  for helping to evaluate
the ``size" ---given in terms of number of elements, weighted means, or similar mathematical features--- of the neighborhoods of the nodes. Together, both tools (metric and measure) allow to define the fundamental object of our model: the density of the family of neighborhoods of  a given node.

The abstract main supporting idea of our model is that  fraud can be detected by searching for \textit{unusual concentration of mass phenomena} in a specifically defined topological graph. Broadly speaking, it can be established in the following terms: \textit{ the ``map of density" of a graph should follow an easy-to-recognize pattern. If no previous information on the pattern is available, then the hypothesis must be that the relevant vertices ---the ones that must focus the attention of the anti-fraud analysts---  are the ones in which there is an anomalous density distribution. In other words,  the uniform density distribution is assumed as reference pattern. Small local densities as well as big local densities should indicate a ``hot node" in terms of corruption, and would allow to classify the different schemes of fraud. }

In this article ---of mathematical nature--- we firstly present the mathematical structure, showing at each step examples that would help the reader to follow the development of the model. The main ideas will be shown in the central part of the paper. Some examples and applications are explained in the final part.

Let us introduce some technical formal concepts. We use standard mathematical notation.  We will construct our models by starting with a set $\Omega$ of entities, that will be considered as the vertices of a complete graph. The edges of the graph will be weighted for the definition of a metric in it.
We will write $\mathbb R^+$ for the set of non-negative real numbers.
A quasi-pseudo-metric  on a set $\Omega$ (\cite{kunzi,reilly}) is a function $d: \Omega \times \Omega \to \mathbb R^+$ satisfying that for $a,b,c \in \Omega$,
\begin{enumerate}
\item $d(a,b)=0$ if  $a=b$, and
\item $d(a,b) \leq d(a,c)+d(c,b)$.
\end{enumerate}

Such a function is enough for defining a topology by means of the basis of neighborhoods that is given by the open balls. If $\varepsilon >0$, we define the ball of radius $\varepsilon$ and center in $a \in \Omega$ as
$$
B_\varepsilon(a):= \Big\{ b \in \Omega: d(a,b) < \varepsilon \Big\}.
$$
Note that this topology is in fact given by the countable basis of neighborhoods provided by the balls $B_{1/n}(x)=\{ y \in X: d(x,y) \le 1/n \}$, $n \in \mathbb N$.
The resulting metrical/topological structure $(\Omega,d)$ is called a quasi-pseudo-metric space.

If the function $d$ is symmetric, that is, $d(a,b)=d(b,a)$, then it is called a pseudo-metric. If $d$ can be used for separating points ---that is, if $d(a,b)=0=d(b,a)$ only in the case that $a=b$--- but it is not necessarily symmetric, then it is called a quasi-metric. Finally, if both requirements hold ---symmetry and separation---, $d$ is called a metric (or a distance). In this case, the topology generated by the balls is Hausdorff.  These notions have been already used in several applied contexts; let us mention for example  the design of semantic computational tools
(\cite{romaguera,valero})
or the analysis of complexity measures in theoretical computer science (\cite{garcia,garcia2}).

\section{Mathematical structures for detection of fraud in public administration and business.}

In this section we introduce the general framework for understanding the fraud processes into a mathematical structure.  Although our objective is to construct a model based on metrics, our goal is to
open the door to the possibility of applying reinforcement learning tools for fraud analysis. Several researchers have recently used machine learning methods for financial fraud detection (see
\cite{ab,wh,y}). Although we have used some ideas from these documents and related ones (see also references in them), our techniques are new, and we are not yet introducing these artificial intelligence tools into this document.  This task will be the next step in our research program.

As we said in the introduction, we will mix for our model a basic graph structure together with some topological tools, that are introduced by means of a quasi-pseudo-metric. In our formalism, in principle the
graphs used are assumed to be complete, but this is not a restriction: we can assign  weights to the edges, so we can ``almost cancel" relationships by using very large distances between vertices.
Graph-based constructions have been already used for fraud detection, although without the explicit introduction of metric elements (see for example \cite{ako,ebe,hoo}). The way we introduce the metric and its role in the model is our main contribution.

  Let  $\Omega$ be a set of objects of the same class related to the representation of individuals  of a system.
Of course, there is a lot of different ways of defining a metric in a set depending on the supplementary
structure that the set is assumed to have (\cite{deza}).
 Typically, the set of entities in our fraud detection model  is represented by vectors containing information of different type, each class in each coordinate. A vector $v$ in this class (belonging to a subset $\Omega$ of a vector space $V$) is univocally associated to an individual: for example, the set $\Omega$ may be composed by invoices of a given year paid by a public administration. Figure \ref{fig1} shows an scheme of the graph  with Neo4j; although the graph represented there is not complete, it is assumed to be complete for the computation of the distance. Each vector may be given by the attributes  of the invoice, for example,
First coordinate= date of payment,
Second coordinate= total amount paid,
Third coordinate= name of the company,
that is,
$$
v=
\big( \text{date of payment, }
\text{total amount paid,}
\text{ name of the company}
\big).
$$
\vspace{.1cm}

\begin{figure}[h]
\includegraphics[scale=.39]{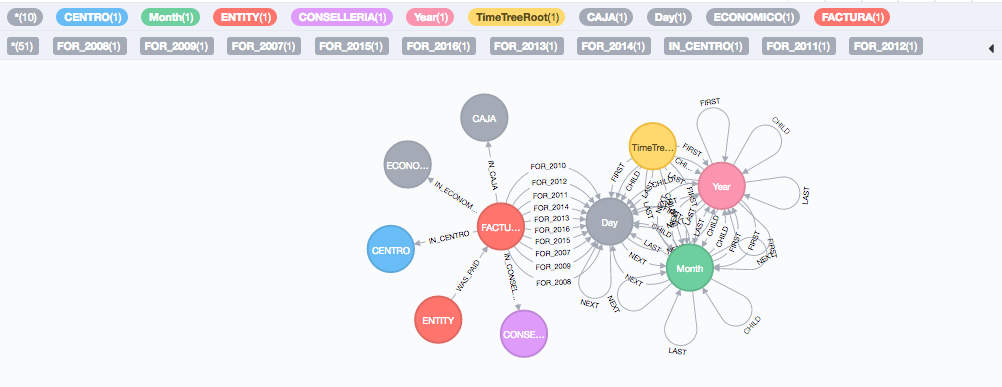}
\caption{An example of a (non-complete) graph for the analysis of invoices in a certain public administration.  }
\label{fig1}
\end{figure}

\vspace{0.2cm}

Let us consider now a quasi-pseudo-metric $d$ in the set $\Omega$.  The explanation of different systematic  procedures for defining it will be given in the next section. In the model it has to represent the proximity of different elements of $\Omega$ among them, and the definition has to make sense for measuring the economic activity (or other kind of relevant activities) related to the process that is being analyzed. For instance, in the previous example a reasonable distance will be given by the following function. Let $v=(x_1,x_2,x_3)$ and $ w=(y_1,y_2,y_3)$ be elements of $\Omega$. We define
$$
d(v,w)=d_1(x_1,y_1) + d_2(x_2,y_2) + d_3(x_3,y_3),
$$
where $d_1(x_1,y_1)= |x_1-y_1|,$  $d_2(x_2,y_2)=|x_2-y_2|$ and  $d_3(x_3,y_3)= 0$ if the invoices $v$ and $w$ were paid to the same company, and $d_3(x_3,y_3)= 1$ otherwise. This clearly defines a distance.

Let us explain other example with some details.

\begin{example} \label{easyex1}
The set of objects $\Omega$ is defined by companies involved in providing services to the public administration in a given year. Each of them is represented by a vector defined by

\begin{itemize}

\item
First coordinate= total amount paid to the company (in K Euros).

\item
Second coordinate= number of services provided by the company.

\item
Third coordinate= geographical location of the company (first coordinate of the position vector).

\item
Fourth coordinate= geographical location of the company (second coordinate of the position vector).

\end{itemize}

This set would be considered by the analyst an ``adequate system", in the sense that it would contain enough information for detecting an anomalous behavior. We identify each company with its representing vector, that is, $\Omega$ is a subset of $\mathbb R^4$. We have to measure the distance among the elements that are considered here. The first obvious choice is to measure the Euclidean distance among vectors, that is if $v_1 ,v_2 \in \Omega$,
$$
d(v_1,v_2)= \Big\| v_1- v_2 \Big\|_2,
$$
where $\big\| \cdot \big\|_2$ denotes the Euclidean norm in $\mathbb R^4.$ However, this option provides an information that only allows to compare companies among them, and grouping them by similarity of activity and location. A priori, it does not seem to be useful for fraud detection.

 A more subtle option would be the following. Consider the seminorms
$$
p_E(x_1,x_2,x_3,x_4)= \Big\| (x_1,x_2,0,0) \Big\|_2, \quad v=(x_1,x_2,x_3,x_4) \in \mathbb R^4,
$$
and
$$
p_L(x_1,x_2,x_3,x_4)= \Big\| (0,0,x_3,x_4) \Big\|_2, \quad v=(x_1,x_2,x_3,x_4) \in \mathbb R^4.
$$

Both of them are seminorms, and so the formulas $d_E(v_1,v_2)=p_E(v_1-v_2)$ and $d_L(v_1,v_2)= p_L(v_1-v_2)$ define pseudo-metrics ($d(v_1,v_2)=0$,  does not necessarily imply $v_1=v_2$). The first one allows grouping companies by similar economic activity ---that is, a small neighborhood of a company/vector $v$ contains companies with similar economic relation with the public administration. Also, a big value of $p_E(v)$ in comparison with the values of $p_E$ of other companies indicates a big economical activity, that would be an indication either of  fraud or risk of fraud. The second one ---$d_L$--- would be used for
detecting changes of names of the same company for hiding an unusual recruitment with the public administration of a  single company.

\end{example}

\vspace{0.5cm}

Let us define now   two more structures. Consider the $\sigma$-algebra $\mathcal B$ of Borel sets of $(\Omega, d)$ ---typically, $\Omega$ will be a finite set and $\mathcal B$ will be the class  $2^\Omega$ of all the subsets of $\Omega$---. Consider a Borel measure $\mu: \mathcal B \to \mathbb R^+$. On the other hand, consider
also a function $\psi: \Omega \times \mathbb R^+ \to \mathbb R^+$ that is increasing with respect to the second variable. It will be considered as a radial weight associated to the radius of the balls for the metric topology.

\begin{definition}
Let
$F(\mathbb R^+, \mathbb R^+)$ be the set of real non-negative functions acting in the positive real numbers.
We define the \textit{density function}  $\mathcal F$ as the function-valued map
$$
\mathcal F: \Omega \times \mathbb R^+ \to F(\mathbb R^+, \mathbb R^+)
$$
 given by
$$
(a, \varepsilon) \mapsto \mathcal F(a, \varepsilon)=f_a(\varepsilon):= \frac{\mu (B_\varepsilon(a))}{\psi(a,\varepsilon)}.
$$
\end{definition}

\begin{remark}  \label{otroex2}
Let us explain a ---in a sense canonical--- example of this notion. Consider a finite set of companies $\Omega=\Omega_0$ in the setting of Example \ref{easyex1}. Take $\mu(\cdot)=| \cdot |$ to be the counting measure on the $\sigma$-algebra of all finite subsets $2^{\Omega_0},$ and  $\psi(a, \varepsilon) = \varepsilon^4$ for all $a \in \Omega_0$ ---the power $4$ for representing the magnitude of a ``volume" in a space of $4$-dimensions---. The metric $d$ is the one defined in the first part of this example. In this case,
$$
\mathcal F(a, \varepsilon)=
f_a(\varepsilon)
$$
$$
= \frac{|B_\varepsilon(a)|}{\varepsilon^4}
= \frac{1 }{\varepsilon^4} \times \, \Big( \text{number of companies in $\{b \in \Omega: \|b-a\|_2 < \varepsilon\} $}\Big).
$$
This formula is clearly defining a density-type parameter: it is given by a ratio among ``number of things" in a given volume of the space and the ``size" of such volume.
\end{remark}

We are prepared now to define the main concept of this paper.

\begin{definition} \label{conmass}
Let $r > 0$.
We define the \textit{concentration of mass } (out of a neighborhood of the element $a$ of size $r$), or the \textit{local density around $a$,} as the function
$C_r:\Omega \to \mathbb R^+ \cup +\infty$ given by
$$
C_r(a)= \int_r^{+ \infty}  f_a( \varepsilon) \, d \nu(\varepsilon), \quad  a \in \Omega,
$$
where $\nu$ is (another countably additive) Borel measure on $(0,\infty)$.
\end{definition}

For $\nu$,  we are thinking for example on a Dirac's delta of a given value $\varepsilon_0 >0$, or Lebesgue measure $d \varepsilon.$ Note that the requirement $r>0$ is imposed to assure the convergence of the integral, at least in the canonical case explained in Remark \ref{otroex2}. In the standard finite case, if $d$ is a distance, it can be taken as the minimum of all the pairwise distances in the set $\Omega$ not being $0$, assuring in this way that $B_r(a)$ contains just an element for any $a \in \Omega$.

\vspace{0.3cm}

The central methodological idea of the present paper is that fraud detection can be considered as a systematic procedure for finding ``outliers" in a quasi-pseudo-metric space.  Indeed,  fraud can be modeled as a \textit{ concentration of mass phenomenon: that is, elements $a \in \Omega$ are associated to processes that are ``suspicious of fraud"  if  $C_r(a)$ has an unexpected  value ---that is, either ``too high" or ``too low" when comparing with the mean value---.} Each of these deviations can be interpreted in different terms, providing different figures of fraud.

It must be taken into account that special elements in the system may have  high values of $C_r$ and this situation can be considered as ``normal": for instance, if there is only one company providing a given service; or, the name of the responsible of the public administration would appear in all the invoices.

Although the way of measuring local density given in Definition \ref{conmass} seems to be the most adequate to the original problem, other ways of measuring this magnitude would make sense. For instance, for the discrete case we can compute the supremum of the  size of the balls $r$ for which the ball contains only its center $a$, that is
$$
r_{\max}(a):= \sup \{r >0: |B_r(a)|=1\},
$$
that coincides with the minimum distance to the closer element of the space, that is
$$
r_{\max}(a) = \min \{d(a,b): b \in \Omega, \, b \ne a \}.
$$
Note that in this case, a big value of $r_{\max}$ means small density.

In the examples in this section it has been used the Euclidean norm in the finite dimensional spaces for constructing the underlying topological structure. This way of measuring the distances is easy and provides directly a metric in the set $\Omega$. However, this is not the best option in general, and an alternate method for defining metric structures is required. The reason is that often the indexes that are naturally used for indicating the distance among elements of $\Omega$ are not metrics; in fact, they are not quasi-pseudo-metrics. Let us explain this relevant point with an example.

Suppose that $\Omega$ is a set of person in a social net, and we have a function $\phi$ that ``measures" the ``level of familiarity" among the elements of $\Omega$ in the following way: $\phi(a,b)=1$ if $a$ an $b$ are close friends,  $\phi(a,b)=2$ if $a$ an $b$ are friends but they meet occasionally,  $\phi(a,b)=3$ if $a$ an $b$ are just acquaintances, and  $\phi(a,b)=4$ if $a$ an $b$ never met. It may clearly happen that $a$ is a close friend of $b$, $b$ is a close friend of $c$, but $a$ and $c$ are only acquaintances; that is
$$
3= \phi(a,c) > \phi(a,b)+ \phi(b,c)= 1+1,
$$
and so the triangular inequality does not hold. This means that $\phi$ is not a quasi-pseudo-metric, but a natural function for measuring social distances.

We will solve this problem by defining a general rule for generation of quasi-pseudo-distances by means of the notion of proximity function, that will be introduced in the next section. As we will see there, the function $\phi$ above is a canonical example of such a proximity function.

\section{The general scheme  of graph quasi-pseudo-metrics for fraud detection}

Distances have often been used for graph analysis in different contexts where graph theory is applied. However, this use is made for comparison between graphs, sometimes also for fraud detection. Metrics are defined to measure the distance between two graphs, not to measure distances between vertices within a given graph (see for example
\cite{cha,chu,gao}).
In this section we are interested in defining a general procedure for analyzing relations inside a graph $\Omega$ defined by ``entities"  (including for example  persons or companies) using the information appearing in text documents, considering that as sets of emails, contracts, invoices and so. The way of doing this is to construct a distance in the graph by means of these elements. Since the very beginning of the modern graph theory, the introduction of a metric in the graph for studying its properties has been used as a  relevant tool
\cite{gra,hak}. In our case, the design of the metric is directly related to the application of the model for fraud detection. Some concrete models based on similar ideas have been recently published for particulars aspects
of fraud detection, as financial reporting fraud \cite{gla}.

Our method follows the next steps.

\vspace{0.3cm}

\begin{itemize}

\item[1)] Detection and definition  of a non-ambiguous set of entities for starting the analysis. For doing this, the analyst has to choose it, and a specific setting should be performed for a fixed kind of fraud. Automatic processes can also be used: for example, semantic parsing techniques provided by the Stanford group  could be applied as well as neural networks for training the searching system (see the references for example \cite{stan}).

\vspace{0.3cm}

\item[2)] Definition of the matrix associated to a \textit{proximity function.}  This is a function $\phi: \Omega \times \Omega \to \mathbb R^+$ that describes by means of a non-negative real number a
relation among the entity $a$ and the entity $b$, both of them in $\Omega$, which represent how far the individuals ---``entities"--- are connected as elements of the network concerning the economical/administrative activities. A small value of $\phi(a,b)$ means that both $a$ and $b$ can be often found as parts of the same activity/business; a big one, that there is not such a relation. Although the function is supposed to be bounded (typically, by $1$), it is not assumed to be a distance. However, it may be assumed to be symmetric and $\phi(a,b)=0$ if and only if $a=b$, and so it only fails subadditivity for being a metric; such functions are sometimes called semimetrics.

\vspace{0.3cm}

\item[3)] Definition of a distance on the set $\Omega$ by using a ``triangular gauge" for $\phi$, that is, a new function $d: \Omega \times \Omega \to \mathbb R^+$ that satisfies that
\begin{itemize}
\item[a)] it is a metric,

\item[b)] and for all $a,b \in \Omega,$ $d(a,b) \le \phi(a,b)$.

\end{itemize}
Of course, for this to be true we need a proximity function $\phi$ that is symmetric and separates points.
In particular, $d(a,b)=0$ if and only if
$a=b$.

We will explain later on how to define explicitly such a function $d$ given a function $\phi$. In fact, the method that we propose is the main contribution of the present work, and has been performed in a specific way for solving the problem that we explained above and we originally faced.

\end{itemize}

\subsection{The triangular gauge of a proximity function $\phi$.}  \label{forcanon}
For the construction of such a gauge, given a function $\phi$ with the requirements explained above we use a path-distance-like definition
by considering a path distance in the global graph $\Omega$, in which all the
vertices are assumed to be  connected ---a complete graph---. We analyze the properties of such a metric, and several procedures for defining the initial proximity matrix $( \phi(a,b) )_{(a,b) \in \Omega \times \Omega}.$

In Section 15.1 in \cite[p.276]{deza}, a weighted path metric for a connected graph is defined as follows. If $e$ is an edge of the graph, write $w(e)$ for the value of a positive weight; $w$ is so assumed to be a (real positive) function acting in the set of edges of the graph.
The path distance $d_G$ among to vertices $a$ and $b$ of the graph is given by
$$
d_G(a,b):= \inf  \Big\{\sum_{e_i \in P} w(e_i) \Big\},
$$
where the infimum is computed over all paths $P=\{e_i:i \in I_P\}$ that allow to go from $a$ to $b$.

We are interested in a construction that is similar to (but not equal to) a weighted path metric defined on the set of all the vertices of a connected graph. In our case all couples of elements of the set are assumed to be directly connected by an edge, that is, the graph is complete. Consider a non-increasing sequence $W:=(W_i)_{i=1}^\infty$ of positive real numbers, all of them less or equal to one. Given two points $a,b \in \Omega$, we define
a function acting in $\Omega \times \Omega$ by
$$
d_\phi(a,b)= \inf \bigg\{ W_1 \phi(a,b), \, \inf \Big\{ W_2 \big( \phi(a,c)+ \phi(c,b) \big): a \ne c \ne b, \, c \in \Omega \Big\}, ...
$$
$$
...,
\inf \Big\{ W_n \big( \phi(a,c_1) + \sum_{i=1}^{n-2} \phi(c_i,c_{i+1})+ \phi(c_{n-1},b) \big), \, a \ne c_1, \, c_i \ne c_{i+1}, \, c_{n-1} \ne b  \Big\}, ... \bigg\}.
$$

A simple calculation shows the next result.

\begin{proposition}
The function $d_\phi(a,b)$ defined above is a pseudo-metric on $\Omega.$ Moreover, if $\Omega$ is finite
and there is a constant $Q>0$ such that  $ 1/i \le Q \, W_i$ for $i=1,...,\infty,$ then $d_\phi$ is a metric.
\end{proposition}
\begin{proof}
A simple look to the formula shows that $d_\phi$ is symmetric due to the symmetry of $\phi.$
Let us show the triangular inequality. Take $a,b,c \in \Omega$ and fix $\varepsilon >0.$ Suppose that the infimum in $d_\phi(a,b) $ and  $d_\phi(b,c) $ is attained ''up to $\varepsilon >0$" for
$$
W_{n_0} \big( \phi(a,c_1) + \sum_{i=1}^{n_0-2} \phi(c_i,c_{i+1})+ \phi(c_{n-1},b) \big), \, a \ne c_1, \, c_i \ne c_{i+1}, \, c_{n_0-1} \ne b
$$
and
$$
W_{n_1} \big( \phi(b,c'_1) + \sum_{i=1}^{n_1-2} \phi(c'_i,c'_{i+1})+ \phi(c_{n-1},c) \big), \, b \ne c'_1, \, c'_i \ne c'_{i+1}, \, c'_{n_1-1} \ne c,
$$
respectively. Now take $n_2= n_0+n_1-1$ and  the sequence $a,c_1,...,c_{n_0-1},b,$ $c'_1,...,c'_{n_1-1},c,$ that satisfies the requirement  that each element is different from the previous one. It contains $n_2$ elements, and so using the fact that
$$
W_{n_2} \le \min \{W_{n_0}, W_{n_1} \}
$$
 we have that
$$
d_{\phi}(a,c)  \le d_{\phi}(a,b) + d_{\phi}(b,c) + 2 \varepsilon.
$$
Since $\varepsilon >0$ is arbitrary, we get that $d_\phi$ satisfies the triangular inequality.

For the second statement, just note that, since $\Omega$ is finite and $\phi$ separates points, we have that there is a constant $k $ such that
$$
0< k < \phi(a,b) \quad \textit{for all} \,\,\, (a,b) \in \Omega \times \Omega.
$$
Therefore, if $a \ne b$ we have
$$
d_\phi(a,b) \ge  \inf \{  k \cdot W_i \cdot i: i \in \mathbb N\} \ge  \frac{k}{Q} >0.
$$
This proves that $d_\phi$ is indeed a metric.

\end{proof}

In what follows we will use the particular case given by the weights sequence $W=(1/i)_{i=1}^\infty,$ and so the distance function is defined by
$$
d_\phi(a,b)= \inf \bigg\{  \phi(a,b), \,\,\, \inf \Big\{\frac{ \phi(a,c)+ \phi(c,b) }{2}: a \ne c \ne b \Big\},
$$
$$
\inf \Big\{ \frac{  \phi(a,c_1)+ \phi(c_1,c_2) + \phi(c_2,b)  }{3}: a \ne c_1 \ne c_2 \ne b \Big\}, \,
$$
$$
...,
\inf \Big\{ \frac{  \phi(a,c_1) + \sum_{i=1}^{n-2} \phi(c_i,c_{i+1})+ \phi(c_{n-1},b) }{n},  a \ne c_1,  c_i \ne c_{i+1},  c_{n-1} \ne b  \Big\} ... \bigg\}.
$$

Note that for computing this infimum we have to deal with an infinite set of numbers. However, if $|\Omega|$ is finite we can approximate the distance by restricting the previous formula to the first $n$ terms appearing in the infimum  ---approximation of order $n$---. The following scheme shows the procedure to define an approximation of order 2 to the metric matrix of the model using the formula above.

\vspace{0.2cm}

\begin{algorithm}
\begin{algorithmic}[1]

\STATE   Fix a set $\Omega=\{a:  \textit{$a$ is an entity in the fraud model} \} \ne \emptyset$.

\REQUIRE  $ 2 \le |\Omega| < \infty$

\WHILE {$a_k \in \Omega$}

\STATE
For $i \in \{1,\cdots,|\Omega|\}, \,\, \textit{compute} \,\, a_i^k =\phi(a_k,a_i)\in \mathbb R^+.$

\STATE Define $\,\, \Phi^k = (\phi(a_k,a_i))_i.$


\STATE $D(1)_\phi(a_k,a_i) \leftarrow \phi(a_k,a_i).$

\STATE $D(2)_\phi(a_k,a_i) \leftarrow  \min \big\{ \frac{\phi(a_k,a_l) + \phi(a_l,a_i)}{2}: a_l \in \Omega\big\}.$

\STATE $d(2)_\phi(a_k,a_i) \leftarrow \min \{ D(1)_\phi(a_k,a_i), D(2)_\phi(a_k,a_i) \}.$


\STATE Define $D^k(2) = (d(2)_\phi(a_k,a_i) )_i.$

\ENDWHILE

\STATE Define $D(2) = (D^k(2))_k.$

%
%

\end{algorithmic}
\caption{Computation of the order 2 approximation $D_\phi(2)$ to the metric for the graph of the fraud model.}\label{alg:algoritmRM}
\end{algorithm}

\vspace{0.2cm}

Suppose that the set $\Omega$ is finite, $| \Omega|=n \in \mathbb N$. Then we can represent $\phi$ by means of the matrix of its range, that is,
$$
\Phi=
\begin{bmatrix}
  \phi(a_1,a_1) & \cdots & \phi(a_1, a_n) \\
  \vdots & \ddots &  \vdots \\
  \phi(a_1, a_n) & \cdots & \phi(a_n, a_n)
 \end{bmatrix} =
\begin{bmatrix}
  0 & \cdots & \phi(a_1, a_n) \\
  \vdots & \ddots &  \vdots \\
  \phi(a_1, a_n) & \cdots & 0
 \end{bmatrix}.
$$

We will call the matrix $\Phi$ the proximity matrix associated to $\phi$.

\begin{example} \label{noeasyex}
Let us give some examples of proximity matrices.

\begin{itemize}

\item[1)]
The first easy example is given by the metric defined in Example \ref{easyex1}. In this case, the proximity function is just the Euclidean metric; that is, $\phi=d$. Consequently, the corresponding proximity matrix $\Phi$ is a metric matrix.

\vspace{0.2cm}

\item[2)]
Let us show two examples of such
construction that are not defined as in Example \ref{easyex1}.
For the first one,  consider $\Omega$ to be a group of individuals that are involved in a business, and the only information we have about it is written in a set  $M$ of documents (see Figure \ref{fig2}).
We want to design an analysis of the influence of the individuals in $\Omega$ in the business. In order to do this and as a first approximation, we consider the following proximity function.

Given $a,b \in \Omega$, take the number  of times $M_{a,b}$ that $a$ appears together with $b$ in a document. Define
$$
\phi_M(a,b)= \frac{ M-M_{a,b}}{M}, \quad a,b \in \Omega.
$$
Another step is needed to clean the matrix in case there are two different individuals in $\Omega$ such that they coincide in all the documents. In this case, they must be considered just as only one vertex of the corresponding complete graph. Note also that $M_{a,b}=1$ indicates that $a$ and $b$ are not appearing together in any document. However, this does not mean that the distance among them has necessarily the maximum value. The reason is that it may happen that $a$ appears in a document with $c$, and $c$ with $b$.
Using an adequate formula for $d_\phi$ ---for example the one given by the weights $W_i=1/i$ as in the particular case given above---, we can easily see that $d_{\phi_M} (a,b) < 1$.

\vspace{0.2cm}

\item[3)]  Let us show now a different way of defining a proximity function for the same problem.
Let $N= |\Omega|$ and assume that there are $M$ documents.  Take the $N \times M$-matrix $C$  of all the counts $C(a,m)$ of the times that the individual $a$ appears in document $m$. Normalize all the vectors appearing in the rows and compute $A=C \cdot C^T.$ It is an $N \times N$-matrix giving the ``cosine"  between elements of $\Omega$. If the element $A(a,b)$ is near to one, this means that they appear in almost the same documents; if it is near to $0$, it means that they are not appearing together.

Take the $N \times N$-matrix $\mathbb I_{N \times N}$ in which all the coefficients are equal to $1$, and compute $\Phi$ as
$$
\Phi= \mathbb I_{N \times N} - A.
$$
It  gives a different proximity matrix. Actually, this construction is the one that we will consider as standard, and will be developed with some detail in the next section. As we will show there, it can be interesting to combine different metrics, some/all of them defined by proximity functions.
\end{itemize}

\begin{figure}[h]
\includegraphics[scale=.5]{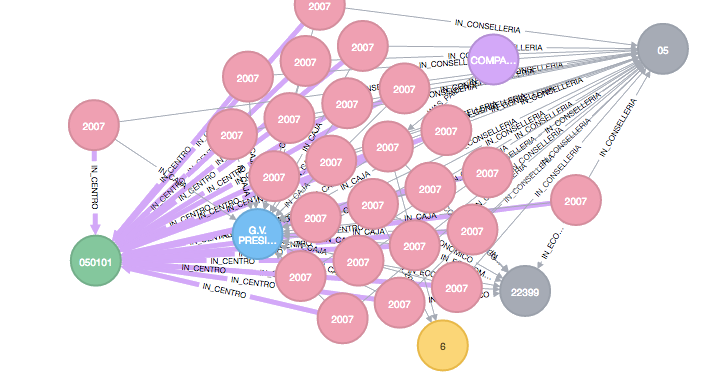}
\caption{``Hidden" representation of a model for fraud detection with no explicit labels for identifying the entities involved (Neo4j).
}
\label{fig2}
\end{figure}

\end{example}

\subsection{ Proximity functions defined by means of correlation matrices: the standard model.}

Let us fix a canonical version of the formulae explained in the previous parts of this section. It follows the lines of Example \ref{noeasyex}, 3).

\vspace{0.2cm}

\begin{itemize}

\item[A.]
Take a set of $N$ entities $\Omega$ and a set of $M$ properties ---quantifiable by means of positive real numbers--- associated to each element $a \in \Omega$.  Construct the set of $N$ vectors $v_a$ each of them containing the numerical value of the properties of a fixed $a \in \Omega$.

\vspace{0.2cm}

\item[B.]
Take the matrix $C$ defined in a way that  each row is such a vector $v_a$ after normalization, that is $v_a/\|v_a\|_2$ (we use the Euclidean norm for normalizing).

\vspace{0.2cm}

\item[C.]
Consider the correlation matrix  $A= C \cdot C^T$ and take  as  proximity matrix $\Phi= \mathbb I_{N \times N} - A.$ Note that it is symmetric.

\vspace{0.2cm}

\item[D.] Define the pseudo-metric $d_\phi$.

\vspace{0.2cm}

\item[E.] The final distance for performing the analysis is given by the formula
$$
d(a,b)= k \, \cdot \,\frac{\| v_a - v_b\|_2  }{\max \{\|v_c\|_2: c \in \Omega\}} + d_\phi(a,b), \quad a,b \in \Omega.
$$
Here, $k >0$ is a parameter for balancing both components of the distance. The first one allows to measure the size of the vectors, for detecting the case that one of its values has unexpected values (for example, a big ammount of money appearing in any coordinate of $v_a$). The second one provides information about the coincidence of coordinates, measuring it using the ``cosine distance".

\end{itemize}

\vspace{0.2cm}

Let us explain a complete example using this method.

\begin{example}
Consider $4$ companies, $a_i$, $i=1,...,4$, which have been hired by a public administration (PA) for doing similar services. We are interested in analyzing if there is any irregular behavior in any of them in 2017. We will show two problems and the models that correspond to each of them. We only have information regarding total amount of money that PA paid to each of them in 2017 and the number of contracts with each company.

\vspace{0.2cm}

\begin{itemize}

\item[(1)] \textit{Suppose that we want to analyze if the   total amount of money $x_i$, $i=1,...,4$, got by each company $a_i$ is either equally distributed among all the companies or we can find   different patterns  regarding that to divide the companies in two groups.} Let us use the procedure explained above. The ``vector of properties" $v_i$ for each company $a_i$ contains just a coordinate, $x_i$. The values (in thousands of euros) are $x_1=4$, $x_2=2$, $x_3=2$, and $x_4=1$. The ``Euclidean part" of the pseudo-distance is then given by
$$
d_E(a_i,a_j):= | x_i - x_j |/  \max \{4,2,1\} = | x_i - x_j |/4, \quad i,j =1,...,4.
$$
The part of the pseudo-metric given by the correlation matrix is given (after normalization) by the trivial formula
$$
\mathbb{I} - A= \mathbb{I} - C \cdot C^T= \mathbb{I} -
\begin{bmatrix}
  1 \\
 1 \\
  1 \\
1
 \end{bmatrix} \cdot
\begin{bmatrix}
  1 & 1 & 1  & 1
 \end{bmatrix}
=
\begin{bmatrix}
  0 & 0 & 0 & 0 \\
0 & 0 & 0 & 0 \\
 0 & 0 & 0 & 0 \\
0 & 0 & 0 & 0
 \end{bmatrix}.
$$
Thus, the final pseudo-metric contains only the Euclidean component, and is represented by the matrix
$$
d= d_E = \begin{bmatrix}
  0 & 1/2 & 1/2 & 3/4 \\
1/2 & 0 & 0 & 1/4 \\
 1/2 & 0 & 0 & 1/4 \\
3/4 & 1/4 & 1/4 & 0
 \end{bmatrix}.
$$
This pseudo-metric allow to separate the set of the four companies in two disjoint balls; indeed, for example
for $\varepsilon= 3/8$, we have
$$
B_{3/8}(a_1)= \{a_1\}, \quad \text{and} \quad B_{3/8}(a_2)= \{a_2,a_3,a_4\}.
$$
The local density in both companies, computed as the ratio among the number of elements in each ball and the radius of the ---one dimensional--- balls give the values for $\varepsilon = 3/8$,
$$
\text{Density}_{3/8}(a_1) = | B_{3/8}(a_1)|/(3/8) = 8/3
$$
and
$$
Density_{3/8}(a_2)=|B_{3/8}(a_2)|/(3/8)= 8
$$
$$
= Density_{3/8}(a_3)= Density_{3/8}(a_4).
$$
Therefore, it can be easily seen that there is a concentration of mass around $a_2$, and $a_1$ is surrounded by an area of low density. In this sense, it can be established that $a_1$ is an isolated point in terms of density, so it is suspicious of receiving an special treatment from PA. Of course, this fits with the fact that $a_1$ got the biggest amount of money in the contracts among all companies, and the difference with the other ones seems to be meaningful.


\vspace{0.2cm}

\item[(2)] Suppose now that we want to analyze a different aspect of the same problem, and we include in the investigation the number of contracts of each of the companies with PA in 2017 given the total amounts of money presented in (1). Now we consider two properties ---two-coordinates  vectors--- for each company: the first coordinate is the amount of money in (1), and the second one if the number of contracts. We have the following values: $a_1= (4,3)$, $a_2=(2,1)$, $a_3=(2,2)$, and $a_4=(1,1)$. For the aim of simplicity, we identify the companies $a_i$ with its two-coordinates property vectors $(x_i,y_i)$, $i,j=1,...,4$.

As in the previous case, we have that the Euclidean part of the distance is given by the Euclidean norm divided by the maximum of the norms, that is, taking into account that
$$
\|a_1\|= 5, \,\,\,\,\,\,\, \|a_2\|= \sqrt 5, \,\,\,\,\,\,\, \|a_3\|= 2 \sqrt 2, \,\,\,\,\,\,\, \|a_4\|= \sqrt 2,
$$
we get
$$
d_E(a_i,a_j)= \|(x_i,y_i)-(x_j,y_j)\|_2 / \max\{ \|a_i\|_2\} = \frac{\|(x_i,y_i)-(x_j,y_j)\|_2}{5}.
$$
This gives the metric matrix
$$
D_E=
\begin{bmatrix}
 0 &  \frac{2 \sqrt{2}}{5} & \frac{\sqrt 5}{5} & \frac{\sqrt 13}{5}  \\
\frac{2 \sqrt 2}{5} & 0 & \frac{1}{5} & \frac{1}{5} \\
\frac{\sqrt 5}{5} & \frac{1}{5} & 0 & \frac{\sqrt 2}{5} \\
\frac{\sqrt 13}{5} & \frac{1}{5} & \frac{\sqrt 2}{5} & 0
 \end{bmatrix} \sim
\begin{bmatrix}
 0 &  0.566 & 0.447 & 0.721  \\
0.566 & 0 & 0.2 & 0.2  \\
0.447 & 0.2 & 0 & 0.283 \\
0.721 & 0.2 & 0.283 & 0
 \end{bmatrix} .
$$
On the other hand,
the proximity  matrix given by the correlation matrix is in this case meaningful. Indeed,
$$
\mathbb{I} - A= \mathbb{I} - C \cdot C^T
$$
$$
= \mathbb{I} -
\begin{bmatrix}
  \frac{4}{5} & \frac{3}{5} \\
   \frac{2}{ \sqrt 5} & \frac{1}{ \sqrt 5} \\
    \frac{1}{\sqrt 2} & \frac{1}{\sqrt 2} \\
     \frac{1}{\sqrt 2} & \frac{1}{\sqrt 2}
 \end{bmatrix} \cdot
\begin{bmatrix}
 \frac{4}{5} &
   \frac{2}{ \sqrt 5} &
    \frac{1}{\sqrt 2} &
     \frac{1}{\sqrt 2} \\
 \frac{3}{5} &
  \frac{1}{ \sqrt 5} &
  \frac{1}{\sqrt 2}
   & \frac{1}{\sqrt 2}
 \end{bmatrix}
\sim
\begin{bmatrix}
  0 & 0.016 & 0.010 & 0.010 \\
0.016 & 0 & 0.051 & 0.051 \\
 0.01 & 0.051 & 0 & 0 \\
0.01 & 0.051 & 0 & 0
 \end{bmatrix}.
$$
This is not a pseudo-metric matrix: note for example that
$$
0.051= \phi(a_2,a_3) > \phi(a_2, a_1) + \phi(a_1,a_3) = 0.016+ 0.010.
$$
In order to provide a pseudo-metric $d_\phi$ preserving as much as possible  the size of the coefficients of  the original proximity matrix, we use the formula given in
Section \ref{forcanon} with all weights equal to one, that is $W_i=1$, $i=1,...,4$. We obtain the pseudo-metric matrix
$$
d_\phi \sim \begin{bmatrix}
  0 & 0.016 & 0.010 & 0.010 \\
0.016 & 0 & 0.026 & 0.026 \\
 0.01 & 0.026 & 0 & 0 \\
0.01 & 0.026 & 0 & 0
 \end{bmatrix}.
$$
The final distance matrix is then given by
$$
D= \lambda \, D_E + d_\phi.
$$
This can be used for the analysis in the same way that was made in (1). However, if we look at the two matrices separately, we get more information about the problem.

\begin{itemize}
\item[(i)]
Using $d_E$, we find again a similar conclusion as the one we got in (1): the first company is the only element in the ball of radius $\varepsilon = 0.4$. However, a ball of the same size $\varepsilon = 0.4$ centered in $a_2$ contains the rest of the elements, $a_2, a_3 $ and $a_4$. The same argument that was used in (1)  using $Density_{0.4}$ provides the same conclusion as in (1).

\item[(ii)]  The second matrix ---associated to $d\phi$--- centers the attention in other element. In this case, the ball $B_{0.015}(a_2)$ only contains $a_2$. However, the ball $B_{0.015}(a_1)$ contains $a_1,$ $a_3$ and $a_4.$ The density around $a_2$ is then smaller than density around $a_1$, $a_3$ and $a_4$. This means that $a_2$ would be suspicious of getting a special treatment, or at least that its hiring pattern is not the same. Note that this pseudo-metric measures the proportion between amount of money and number of contracts. The result shows that the company $a_2$ is not following the same proportion, what means that the money associated to each contract is different. This may be just by chance, but also would indicate that there is someone interested in manipulating the standard hiring procedure, and so it would be suspicious of fraud.

\end{itemize}

\end{itemize}

\end{example}

\vspace{0.2cm}

\section{Final remarks: applications of the model to detect irregular behavior of elements in a network}

In this section, and to finish the paper, we give some open ideas for applying the tools that we have shown.
We can consider the following problems, which could be solved by  applying  our metric graph structure.

\begin{itemize}

\item The first and canonical one: \textit{given an entity $a \in \Omega$, find the rest of the elements of $\Omega$ that are near }(distance less than $\varepsilon >0$). This is the first step of the neighborhood analysis that allow to compute a density map for searching anomalous behaviors. But it also gives a primary information, providing  the entities that are close to a given one $a$ with respect to the criterium used for the construction of the proximity function.

\item
\textit{Degree of dependence of the ``graph distance" on a single element $a \in \Omega$}: this is the norm of the difference of the submatrix $D_a$ that is obtained by eliminating the row and column associated to $a$ in the distance matrix $D$, and the distance matrix $D(-a)$ that is computed when the set considered is $\Omega \setminus \{a\}$ instead of $\Omega.$ If the value is small, this means that the element $a$ is not relevant for the graph, it is not really connected or it is not giving easy paths for other entities to be connected.

\item \textit{Optimization}: given a vertex $a \in \Omega$ and a subset $S \subset \Omega$, find the element(s) $b$  in $S$ such that $d_\phi(a,b)$ attains its minimum.

\item \textit{A singular-values-type method for determining the classes of equivalence  of entities in the space having the same behavior}, in the sense that they appear in the same documents. We use
the matrix $A$ defined in Example \ref{noeasyex}, 3).  Consider the individuals $a_1$ to $a_n$  and suppose they are appearing in the same documents, and they are the only ones appearing in these documents. Then we can write the
vectors of the matrix $A$ corresponding to these individuals as
$$
1/\sqrt n \, (1,1,...,1,0,...0),
$$
where the coefficient equal to $1$ appears in the $n$ first positions. On the other hand, the other individuals have coefficients that are all of them $0$ in the first $n$ positions (check that, this is a consequence of the construction of $A$ based in the fact that they are appearing in disjoint documents). When the corresponding submatrix is diagonalized, we obtain an eigenvalue that is not zero and other one that is $0$, that has multiplicity   $n-1$. Therefore, there is only one document-appearing behavior, the rest only repeat the behavior of the first individual. Of course, we rarely are going to find this pure behavior, and so we use the ideas of the singular values method for giving the ``almost zero" version.

For doing this, compute the eigenvalues of the matrix $\{\lambda_i: 1 \le i \le m \}.$ Fix $\varepsilon >0$, and take the subspace $S_\varepsilon$ generated by the eigenvectors associated to the  eigenvalues $\lambda_i < \varepsilon.$ Write the equation
$A= U^T \Delta U$ ($U$ is the matrix of change of basis) and compute the vectors $v_a = (0,\cdots,1, \cdots 0)$ representing the elements $a \in \Omega$ that satisfy that  $U v_a$ is in $S_\varepsilon$. This is the set
that can be eliminated from the original set $\Omega$, since they have an equivalent behavior that any of the ones for which $\lambda_i \ge \varepsilon.$

\end{itemize}

\section{Conclusions}

We have presented a new framework for constructing decision support systems for financial anti-fraud analysis. It consists of a graph structure $\Omega$ together with a distance defined on it,
that models the relations among the entities involved in the analysis. We have shown how to define these metrics by means of examples and applications.

Our main methodological hypothesis has also been established. Together with the metric structure, a measure acting in the $\sigma$-algebra generated by $\Omega$ is considered
in order to define a function that allows to measure the density of the neighborhoods of the elements of the model. Our main axiom claims that a (group of) entity(ies) is suspected of committing fraud
 whenever there is an anomalous density --meaningfully bigger or smaller than the mean--- in his neighborhood. Concrete models and examples for explaining this idea are presented.

%
%
%

\end{document}